\documentclass[reqno,12pt,letterpaper]{amsart}
\usepackage{amsmath,amssymb,amsthm,graphicx,mathrsfs,url,mathabx}
\usepackage[usenames,dvipsnames]{color}
\usepackage[colorlinks=true,linkcolor=Red,citecolor=Green]{hyperref}
\usepackage{amsxtra}
\usepackage{subcaption}
\usepackage{wasysym}
\usepackage{tikz-cd}
\usepackage{array,enumerate}
\usepackage[capitalise]{cleveref}
\crefname{prop}{Proposition}{Propositions}

\let\Re=\Real

\usepackage{graphicx,color}

\usepackage{array}
\newcolumntype{P}[1]{>{\centering\arraybackslash}m{#1}}

\def\wrtext#1{\relax\ifmmode{\leavevmode\hbox{#1}}\else{#1}\fi}

\newcolumntype{L}{>{$}l<{$}}

\def\?[#1]{\textbf{[#1]}\marginpar{\Large{\textbf{??}}}}

\let\epsilon=\varepsilon 
\let\phi=\varphi

\newtheorem{thm}{Theorem}

\numberwithin{equation}{section}

\theoremstyle{definition}

\usepackage{braket}

\crefname{lem}{Lemma}{Lemmas}
\Crefname{lem}{Lemma}{Lemmas}

\usepackage{xcolor}
\definecolor{purp}{RGB}{160, 32, 240}

\newcommand{\cEH}{E_{\rm H}}
\newcommand{\cEX}{E_{\rm X}}
\newcommand{\cEHF}{E_{\rm HF}}

\newcommand{\Ran}{\operatorname{Ran}}

\newcommand{\tr}{\operatorname{tr}}
\newcommand{\Id}{I}

\usepackage[margin=1in]{geometry}

\title[Hartree--Fock coercivity for TBG]{Hartree--Fock coercivity for twisted bilayer graphene}
\author{Kevin D. Stubbs}

\author{Maciej Zworski}

\begin{document}

\maketitle

\begin{abstract}
We study
 the minimisers of the Hartree--Fock functional for flat bands in the chiral model of 
twisted bilayer graphene \cite{magic}. Our model is a continuous version of the discrete 
model considered by Becker, Lin and Stubbs \cite{BLS25,SBL25}  and the first result is a direct 
analogue of main theorems of those papers: at half filling of the flat band of multiplicity two, there are exactly two
minimisers, obtained by occupying either one of the two sub-bands.
(A more general version of this result, stressing its topological origins, was presented in \cite[\S A.8]{App}.) The new contribution is the {\em strict} coercivity of the Hessian of the Hartree--Fock functional at the two minimisers. It is a consequence of the difference between the Chern numbers of the component bands and can be optimistically considered as a soft version of a ``spectral gap''.
\end{abstract}

\section{Introduction}
\label{s:intr}

The purpose of this note is to investigate a Hartree--Fock functional in the continuous model of 
interacting electrons in twisted bilayer graphene (TBG).
The non-interacting part is based on Tarnopolsky, Kruchkov and Vishwanath's \cite{magic} chiral limit  of the Bistritzer--MacDonald model \cite{BM},  see \S \ref{s:rev} for a detailed review and references to mathematical results.
The interacting model is the continuous version of the discrete model studied by Becker, Lin and Stubbs \cite{BLS25,SBL25}.
In these works, the discreteness came from considering supercells (large tori) and the consequent discretisation of quasi-momenta.
Here we follow \cite{App} and work on $ \mathbb R^2 $ (rather than supercells) with continuous quasi-momenta.
We refer to that appendix for a detailed presentation of this theory. 

\begin{figure}
\includegraphics[width=17cm]{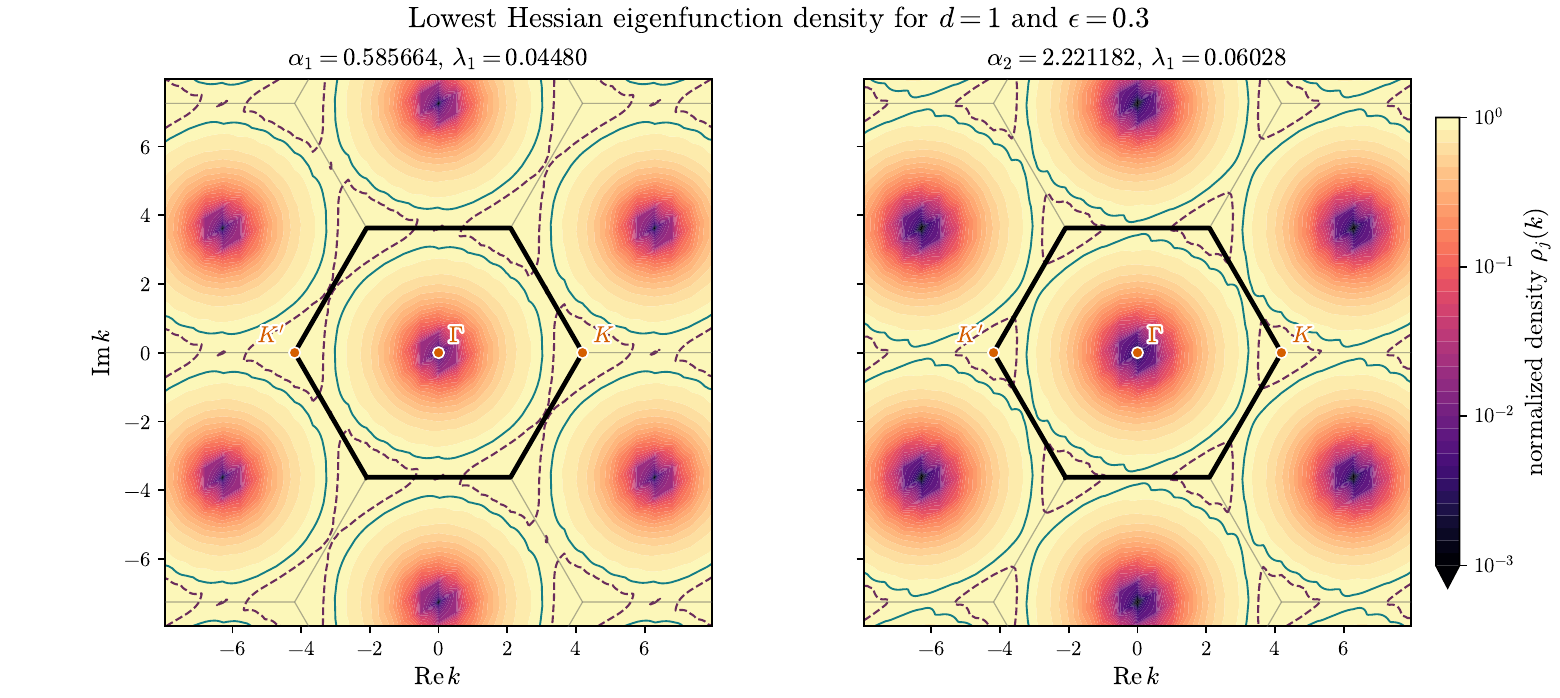}
\caption{Numerically computed lowest eigenfunctions, $ \psi_1 $ (with the eigenvalues $ \lambda_1 $ determining the constants in \eqref{eq:Hess1}) for the first two magic angles of \cite{magic}:
the plot shows the log of the normalised density,
$ \rho_j(k):={|\psi_{1}(k)|^2}/
{\max_q|\psi_{1}(q)|^2}, $ 
for $ \alpha_j $, $ j =1,2 $. 
The parameters in \eqref{eq:Vhat} are $d=1$ and $\epsilon=0.3$. The vanishing at $\Gamma = 0 $ follows from basic symmetries. This figure and Figure~\ref{fig:bdry} were produced, on request, by ChatGPT 5.6.}
\label{fig:hessian-modes}
\end{figure}

We now briefly describe the resulting Hartree--Fock problem.  At a simple magic
angle (see \eqref{eq:mmagic}), the Hamiltonian $H(\alpha)$ reviewed in \S\ref{s:rev} has a flat band of multiplicity two, described by a rank-two Bloch bundle $E=E_+\oplus E_-$.
To study the behavior of this system, we take the strong coupling approach and project the electrons-electron interactions to the flat bands (see Ledwith, Khalaf, Vishwanath \cite{LKV21} for a review).
We use the average subtraction scheme which
determines a one-body subtraction term $P_{\rm sub}$ -- see  \cite[(A.8.10)]{App}.
Thus the effective one-body Hamiltonian and the corresponding many-electron Hamiltonian
are, respectively,
\[
 P:=H(\alpha)-P_{\rm sub},\qquad
 \widehat P+\mathbb V
 =\widehat H(\alpha)+\mathbb V-\widehat P_{\rm sub},
\]
where $ \mathbb V $ is the second quantisation of the interaction term described in 
\cite[\S A.6,\S A.7]{App}.  (We avoid
$ \widehat V $ used there, not to confuse $ \mathbb V $ with the Fourier transform of $ V$.)
The subtraction $\widehat{P}_{\rm sub}$ is meant to remove interactions already incorporated in the
one-body Hamiltonian.
While minimizing the many-electron Hamiltonian is a formidable problem, it was observed in the physics literature -- see Vafek, Kang \cite{KV19}, Bultinck, Khalaf, Liu, Chatterjee, Vishwanath, and Zaletel \cite{Bultinck2020GroundStateHidden}, Bernevig, Song, Regnault, and Lian \cite{BSRL21} -- that the choice of average subtraction produces mathematically clean results.
We refer to \cite[\S A.8, (8.8)--(8.10)]{App} for its
mathematical formulation and to the references above for its motivation and use in interacting
models of TBG. 

A standard physical example of the interaction is the
double-gate screened Coulomb interaction, whose Fourier transform is
\begin{equation}
\label{eq:Vhat}
 \widehat V(\eta)=(2\pi/{\epsilon})
       {\tanh(\tfrac12d|\eta|)}/{|\eta|},
       \ \ \  \epsilon,d>0,
\end{equation}
which is bounded at $\eta=0$ and is $\mathcal O(|\eta|^{-1})$ at infinity; see
\cite[(2.23)]{BLS25} and \cite[(A.8.26)]{App}.

For a periodic one-particle density matrix $\gamma$, the Hartree--Fock energy
per fundamental cell is 
\[
 E_{\rm HF}^\Lambda(\gamma)
 =\tr_\Lambda(P\gamma)
 +\frac12\int_\Omega\int_{\mathbb R^2}V(x-y)
 \left(\rho_\gamma(x)\rho_\gamma(y)
 -\tr\big(\gamma(x,y)\gamma(y,x)\big)\right)\,dy\,dx,
\]
where $ \rho_\gamma(x):=\tr\gamma(x,x)$, $\Omega$ is a fundamental cell of $\Lambda$, and $\tr_\Lambda$ denotes
trace per cell (see \cite[(8.1)--(8.3)]{App}).  After restriction to the flat
band and the above subtraction, this becomes, up to irrelevant constants,
the functional $E_{\rm H}(\Pi)+E_{\rm X}(\Pi)$ described in
\eqref{eq:HFE}.

\begin{figure}
\includegraphics[width=16cm]{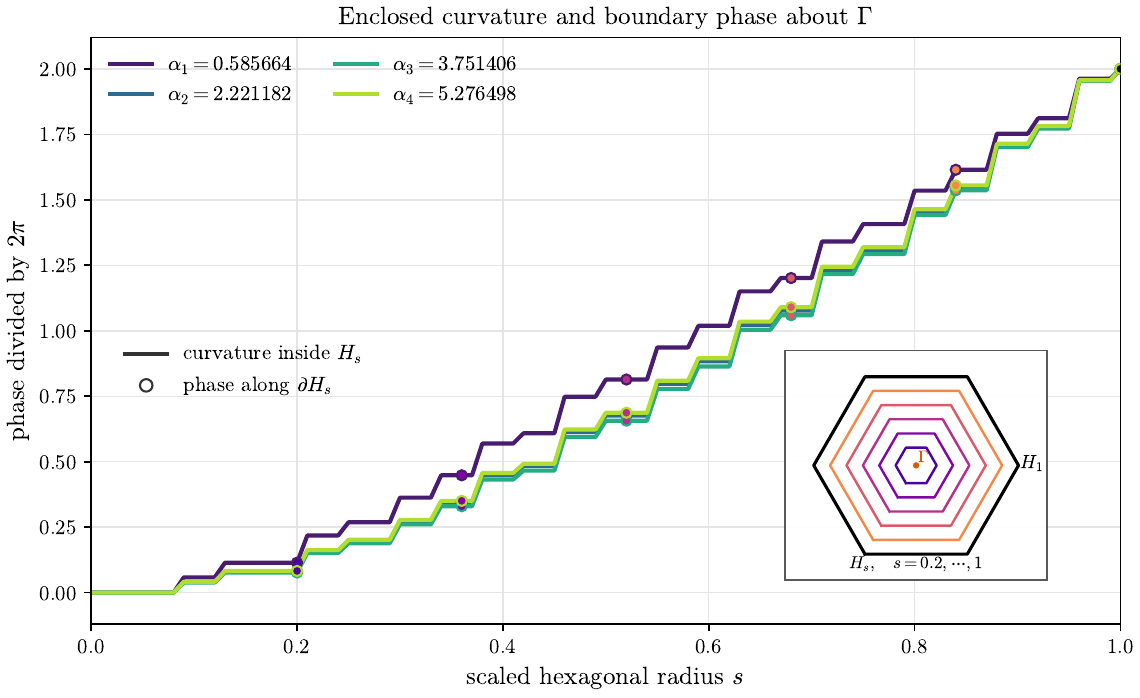}
\caption{Curvature of the tangent line bundle at $ \Pi_+$, $E_-\otimes E_+^*$ ($ \varphi $ in \eqref{eq:Hess1} are essentially sections of it), equipped with its natural connection, for the first four positive magic parameters. Writing $H_s=sH_1$  where $H_1$ is the fundamental hexagon centred at $\Gamma$, the solid curves show the curvature enclosed by $H_s$ (divided by $2\pi$), and the circles,  the corresponding continuously lifted phase along $\partial H_s$.  Since
$
c_1(E_-\otimes E_+^*)=2,
$
all four profiles have the common endpoint $ 2 $. Their intermediate shapes do however depend visibly on  $\alpha$, with the first magic parameter being the most distinct.
}
\label{fig:bdry}
\end{figure}

The flat band has rank two, and a translation-invariant Slater state is
described in Bloch variables by a measurable field of orthogonal projections
$\Pi(k):E_k\to E_k$.  We impose an average half-filling condition:
\[
 \Pi(k)^2=\Pi(k)=\Pi(k)^*,\qquad
 |\mathbb C/\Lambda^*|^{-1}
 \int_{\mathbb C/\Lambda^*}\tr\Pi(k)\,dk=1.
\]
The usual uniform half-filling condition $\tr\Pi(k)=1$ almost everywhere is a
special case.  Allowing the more general condition above is useful in the
continuous setting and is discussed further in \S\ref{s:HFmin}; see also
\cite[\S A.8]{App} and \cite{SBL25}.

Our first result, stated precisely as Theorem~\ref{t:1}, says that the only
minimisers in this larger class are
\[
                         \Pi=\Pi_+\qquad\hbox{and}\qquad\Pi=\Pi_-,
\]
the projections obtained by filling one of the two component flat bands.  In
particular, the minimisers are uniformly half-filled even though uniformity is
not assumed.  This is the continuous counterpart of the results of
\cite{BLS25,SBL25}. Its topological version in \cite[Theorem~8.1]{App} applies
more generally to a transfer-invariant splitting into line bundles with
different Chern numbers.

The new result of this note concerns stability of these minimisers.  On the
Hilbert manifold of rank-one projection fields, Theorem~\ref{t:2} gives, at
either minimiser, the strict coercive estimate
\begin{equation}
 \operatorname{Hess}_{\Pi_\pm}E_{\rm HF}(\varphi,\varphi)
 \geq c \int_{\mathbb C/\Lambda^*}|\varphi(k)|^2\,dk,
 \qquad c >0.
\label{eq:Hess1}
\end{equation}
The mechanism is topological: the component line bundles have different Chern
numbers, which rules out a nonzero null direction for the Hessian.  This can be
viewed as a mean-field, or ``soft'', analogue of a spectral gap, though we insist that it is
not a statement about the spectral gap of the many-body Hamiltonian. 

Figure~\ref{fig:hessian-modes} illustrates \eqref{eq:Hess1} numerically. It displays the lowest Hessian mode for the first two magic parameters: its eigenvalue, $ \lambda_1 $, measures the weakest neutral variation away from the minimiser, while its density shows where that variation is concentrated in quasi-momentum space. Figure \ref{fig:bdry} displays the curvature of the tangent line bundle (the space where $ \varphi $ in \eqref{eq:Hess1} live) for the first four magic parameters, separating its topologically fixed total from its alpha dependent distribution.

We conclude this introduction with additional pointers to the literature.
The positivity of the constrained Hartree--Fock Hessian is the
usual stability condition for a Hartree--Fock state.  In physics, Thouless
\cite{Th60} related this condition to the stability of collective modes in the
random-phase approximation.  Quantitative coercivity estimates appear to be
much less common, and rigorous many-body gap results derived from such
estimates seem rarer still.  In the reduced Hartree--Fock model for defects in an
insulating crystal, Canc\`es, Deleurence and Lewin \cite{CDL08} obtained a global
coercive bound from convexity and an assumed one-particle band gap.  For the
full Hartree--Fock functional of a weakly interacting Fermi gas on a finite
torus, Benedikter, Nam, Porta, Schlein and Seiringer
\cite[Appendix~A]{BNPSS21} proved a global coercive bound using a discrete
kinetic gap.  In the present flat-band problem the projected one-body energy is
constant and supplies no such mechanism: the absence of zero modes, and hence
coercivity, follows instead from the mismatch of the Chern numbers of the two
component bands.  To our knowledge, this topological mechanism for
Hartree--Fock stability is new.  Different notions of gap should also be kept
separate.  Bach, Lieb, Loss and Solovej 
\cite{BLLS94} proved strict separation
of occupied and unoccupied levels of the self-consistent Fock operator for
finite systems, while \cite{BLS25} establishes a charge gap for adding or
removing an electron in the discrete flat-band model.  Instead, the result here 
controls neutral orbital variations at fixed filling and does not by itself
give a spectral gap above the ground-state sector of the many-electron
Hamiltonian.

\medskip
\noindent {\sc Acknowledgements.}
We gratefully acknowledge partial support from the Simons
Foundation through Targeted Grant Award No.~896630, ``Moir\'e Materials
Magic''.
We also thank ChatGPT 5.6 for its advanced editorial and
mathematical assistance and for facilitating the literature review.

\section{Review of the chiral model of TBG}
\label{s:rev}

In this section we review 
the chiral Bistritzer--MacDonald Hamiltonian \cite{BM} studied in
\cite{magic}. For mathematical accounts of its properties see  \cite[\S 2]{Zw24} and \cite[\S 2.1]{BeZw24}
and references given there.  To introduce the model we consider the lattice governing its symmetries:
\[  \Lambda: =\mathbb Z\oplus\omega\mathbb Z, \ \ \ \omega: =e^{2\pi i/3}, \ \ \ 
 \Lambda^*=\frac{4\pi i}{\sqrt 3}\Lambda, \ \ \ K=\tfrac43\pi, \ \ \
\langle z,w\rangle:=\operatorname{Re}(z\bar w).
\]
Here $ \Lambda^* $ is the dual (reciprocal) lattice and $ K $ is a point of high symmetry:
$ \omega K \equiv K \mod \Lambda^* $. 
Let $\Omega^*$ be a fundamental cell of $\Lambda^*$ .

In the coordinates of
\cite[(2.2)--(2.3)]{Zw24} (see \S 2.1 there for the relation to other coordinate descriptions),  the Hamiltonian
$H(\alpha):H^1(\mathbb C;\mathbb C^4)\to L^2(\mathbb C;\mathbb C^4)$ is given by 
\begin{equation}
 H(\alpha):=
 \begin{pmatrix}0&D(\alpha)^*\\ D(\alpha)&0\end{pmatrix},
\ \ \ \ \ 
 D(\alpha):=
 \begin{pmatrix}
 2D_{\bar z}&\alpha U(z)\\
 \alpha U(-z)&2D_{\bar z}
 \end{pmatrix},
 \label{eq:chiralHamiltonian}
\end{equation}
where $ 2 D_{\bar z } =( \partial_{x_1} + i \partial_{x_2})/i  $, $ z = x_1 + i x_2 $. 

The dimensionless constant $\alpha\in\mathbb C$ is proportional to the inverse twisting angle and the
potential, modelling tunnelling $AB'/BA'$ sites is assumed to satisfy
\begin{equation}
 U(z+\gamma)=e^{i\langle\gamma,K\rangle}U(z),\qquad
 U(\omega z)=\omega U(z),\qquad
 \overline{U(\bar z)}=-U(-z),\qquad \gamma\in\Lambda.
 \label{eq:Usymmetries}
\end{equation}
The specific Bistritzer--MacDonald potential \cite{BM} in this normalisation is given by 
\begin{equation}
 U_{\rm BM}(z)=-\frac{4\pi i}{3}
 \sum_{\ell=0}^2\omega^\ell e^{i\langle z,\omega^\ell K\rangle}.
 \label{eq:UBM}
\end{equation}

We next recall the twisted periodicity and the modified Bloch transform -- see 
\cite[\S 2.1]{BeZw24} and \cite[\S 3.1]{Zw24}. For $u\in L^2_{\rm loc}(\mathbb C;\mathbb C^2)$ and
$\gamma\in\Lambda$, define
\begin{equation}
 L_\gamma u(z):=
 \begin{pmatrix}
 e^{i\langle\gamma,K\rangle}&0\\
 0&e^{-i\langle\gamma,K\rangle}
 \end{pmatrix}u(z+\gamma),
 \qquad
 \mathscr L_\gamma:=
 \begin{pmatrix}L_\gamma&0\\0&L_\gamma\end{pmatrix}
 \label{eq:twistedtranslations}
\end{equation}
on $L^2_{\rm loc}(\mathbb C;\mathbb C^4)$.  We use the twisted periodic spaces
\begin{equation}
\begin{split}
 L^2_k(\mathbb C;\mathbb C^2)
 &: =\{u\in L^2_{\rm loc}(\mathbb C;\mathbb C^2):
 L_\gamma u=e^{i\langle k,\gamma\rangle}u,\ \gamma\in\Lambda\},\\
 L^2_k(\mathbb C;\mathbb C^4)
 &: =\{u\in L^2_{\rm loc}(\mathbb C;\mathbb C^4):
 \mathscr L_\gamma u=e^{i\langle k,\gamma\rangle}u,\ \gamma\in\Lambda\},
 \end{split}
\label{eq:twistedspaces}
\end{equation}
with norms taken over a fundamental cell of $\Lambda$.  We also put 
$ H^1_k:=H^1_{\rm loc}\cap L^2_k $.  These spaces depend only on
$[k]\in\mathbb C/\Lambda^*$.

The modified Bloch transform is given, initially for Schwartz functions, by
\begin{equation}
 {\mathcal B}f(k,z):=|\mathbb C/\Lambda^*|^{-1/2}
 \sum_{\gamma\in\Lambda}e^{i\langle z+\gamma,k\rangle}
 \mathscr L_\gamma f(z).
 \label{eq:Blochunitary}
\end{equation}
It satisfies
\begin{equation}
 \mathscr L_\gamma {\mathcal B}f(k,\bullet)={\mathcal B}f(k,\bullet),
 \qquad
 {\mathcal B}f(k+p,z)=e^{i\langle z,p\rangle}{\mathcal B}f(k,z),
 \qquad p\in\Lambda^*.
 \label{eq:transs}
\end{equation}
As in \cite[\S 5.2]{TaZw}, ${\mathcal B}$ extends to a unitary map
\[
 {\mathcal B}:L^2(\mathbb C;\mathbb C^4)
 \longrightarrow L^2(\mathbb C/\Lambda^*;V),
\]
where the parent bundle (see \cite[\S 9.1]{TaZw},\cite[\S A.7]{App}) is
\begin{equation}
 \label{eq:parent}
\begin{gathered}
 V:=\bigl(\mathbb C\times L^2_0(\mathbb C;\mathbb C^4)\bigr)/\!\sim,
 \ \ \ \
 (k,u)\sim(k+p,\tau(p)u),
 \\
 [\tau(p)u](z):=e^{i\langle z,p\rangle}u(z), \ \ \ p\in\Lambda^*.
 \end{gathered}
\end{equation}
The chiral grading gives
\[
 V=V_+\oplus V_-,\qquad
 V_\pm:=\bigl(\mathbb C\times L^2_0(\mathbb C;\mathbb C^2)\bigr)/\!\sim.
\]

On the modified Bloch side we consider $ H_k ( \alpha ) \mathcal B u ( z, k ) =  ( \mathcal B H ( \alpha ) u )( z, k) $, 
where 
\begin{equation}
 H_k(\alpha)=
 \begin{pmatrix}
 0&D(\alpha)^*+\bar k\\
 D(\alpha)+k&0
 \end{pmatrix},
 \qquad
 D(\alpha)+k:H^1_0(\mathbb C;\mathbb C^2)
 \longrightarrow L^2_0(\mathbb C;\mathbb C^2).
 \label{eq:chiralFibers}
\end{equation}
The eigenvalues of $H_k(\alpha)$ are symmetric with respect to the origin and we write them as
\[  \cdots \leq - E_2 ( \alpha, k ) \leq - E_1 ( \alpha, k ) \leq  0\leq E_1(\alpha,k)\leq E_2(\alpha,k)\leq\cdots. 
\]
A parameter $\alpha$ is \emph{magic} if $E_1(\alpha,k)=0$ for every $k$, and
its chiral multiplicity is
\begin{equation}
 m(\alpha):=\min\{j>0:\max_{k\in\mathbb C/\Lambda^*}
 E_{j+1}(\alpha,k)>0\}.
 \label{eq:mmagic}
\end{equation}
(We remark that $ E_{1} ( \pm K, \alpha ) = 0 $ for all $ \alpha $ and the positivity condition can be replaced by 
the existence of $ k \neq \pm K + \Lambda^* $ such that $ E_{j+1} ( \alpha, k ) = 0 $ -- see \cite[Theorem 1]{BHZ25}.)

For a magic $\alpha$, \cite[Theorem~1]{BHZ26} shows that, for every $k$,
\begin{equation*}
 m(\alpha)=\dim\ker(D(\alpha)+k),\qquad
 \dim\ker H_k(\alpha)=2m(\alpha).
\end{equation*}
Thus a magic angle is called \emph{simple} when $m(\alpha)=1$. In that case,  the corresponding
flat band of the full Hamiltonian has multiplicity two.  Magic angle
multiplicities and their restrictions are studied in \cite{BHZ26,IN25}.
For the potential \eqref{eq:UBM}, Watson and Luskin \cite{WaLu21} proved the
existence of a real magic parameter in $(0.57,0.61)$.  Becker, Humbert and
Zworski \cite[Theorem~3]{BHZ23} gave a different proof, locating the smallest
positive magic parameter in $(0.583,0.589)$ and proving that it is {\em simple}.
In particular, the simple-magic-angle hypothesis used below is non-empty.

At a simple magic angle $\alpha$, in the notation of \cite{BeZw24},
\begin{equation}
 \ker_{H_0^1} \bigl(D(\alpha)+k \bigr)=\mathbb C u(k),
 \qquad
 \ker_{H^1_0} \bigl(D(\alpha)^*+\bar k \bigr)=\mathbb C u^*(k).
 \label{eq:TBGkernels}
\end{equation}
These kernels define smooth line subbundles $E_+\subset V_+$ and $E_-\subset V_-$,
with orthogonal projections
\begin{equation}
 \Pi_\pm(k):V_{\pm,k}\longrightarrow E_{\pm,k}.
 \label{eq:defPEpmk}
\end{equation}
We also put
\begin{equation}
 \Pi_E(k):V_k\longrightarrow E_k:=E_{+,k}\oplus E_{-,k},\ \,\ 
 \Pi_E(k)=\begin{pmatrix}\Pi_+(k)&\ \,0\\ \ \,0&\Pi_-(k)\end{pmatrix}.
 \label{eq:defPk}
\end{equation}
The splitting $E=E_+\oplus E_-$ is canonical.  Moreover,
\begin{equation}
 c_1(E_+)=-1,
 \qquad
 c_1(E_-)=1;
 \label{eq:TBGchern}
\end{equation}
see \cite[Theorem~4]{BHZ25}.  Hence $c_1(E)=0$, and $E$ is a trivial rank-two
bundle over $\mathbb C/\Lambda^*$, although we do not choose a trivialisation until
Section~\ref{s:scHF}.

The chiral and antiunitary symmetries,  respectively (see \cite[(2.9),(2.10)]{BeZw24}),  give
\begin{equation}
 u^*(k)=\mathscr Q u(k),
 \qquad
 u(-k)=\rho(k)\mathscr H u(k),
 \qquad |\rho(k)|=1.
 \label{eq:TBGsymmetries}
\end{equation}
In particular, if
$M_pv(z)=e^{i\langle p,z\rangle}v(z)$ and
$A_k(p):=\langle u(k), M_pu(k)\rangle$, then
\begin{equation}
 A_{-k}(p)=\overline{A_k(p)}.
 \label{eq:Asymmetry}
\end{equation}
We will use \eqref{eq:TBGchern} in the topological part of the proofs below and
\eqref{eq:Asymmetry} to verify a neutrality identity \eqref{eq:neutrality-new} needed to 
establish existence and uniqueness of minimisers (see Theorem \ref{t:1}).

\section{Many body interaction projected to the flat band}
\label{s:flat}

\subsection{Interaction projected to the flat band} 
\label{s:intp} 
We recall from \cite[\S A.7]{App} that the interaction for the $ N$-body system is 
 described using the multiplication operator $M_\eta$, 
\begin{equation}
 \label{eq:s7density}
 M_\eta u ( x ) := e^{ i \langle x, \eta \rangle } u ( x ) 
 \end{equation}
projected to the flat band. We first  describe it  on the modified Bloch transform side. 
For $k\in\Omega^*$,  denote the chosen
representative of $[k-\eta]\in\mathbb C/\Lambda^*$  by $ k_\eta \in \Omega^* $, and write
\begin{equation}
\label{eq:s7keta}
   k-\eta=k_\eta+p_\eta(k),\qquad p_\eta(k)\in\Lambda^*.
\end{equation}
Multiplication by $e^{i\langle\eta,x\rangle}$ changes
$e^{-i\langle k,x\rangle}u(x)$ into
$e^{-i\langle k-\eta,x\rangle}u(x)$.  As in \cite[\S A.7]{App} It induces the intrinsic
parent-bundle map
\begin{equation}
 T_\eta(k):V_k\longrightarrow V_{k_\eta},
 \qquad
 T_\eta(k)[k,u]:=[k-\eta,u].
\label{eq:parentt}
\end{equation}
In the representatives chosen in \eqref{eq:s7keta}, $ T_\eta(k)[k,u]
 =[k_\eta,e^{-i\langle p_\eta(k),x\rangle}u]$, or 
equivalently,
\[   ({\mathcal B} M_\eta{\mathcal B}^{-1}s)(k_\eta) =T_\eta(k)s(k). \]

By \eqref{eq:defPk}, the smooth subbundle $ E_+ \oplus E_- \subset V$ defines the
orthogonal projection
\[
\Pi_E:={\mathcal B}^{-1}\Pi(\bullet){\mathcal B}, \ \ \ \
(\Pi(\bullet)s)(k): =\Pi(k)s(k).
\]
The Bloch
representative of $ \Pi_E M_\eta \Pi_E $ is 
\begin{equation}
 L_\eta:={\mathcal B}\Pi_{E}M_\eta\Pi_{E}{\mathcal B}^{-1}
 :L^2(\mathbb C/\Lambda^*;E)\longrightarrow L^2(\mathbb C/\Lambda^*;E).
\label{eq:s7Leta}
\end{equation}
Explicitly, for $ s \in L^2 ( \mathbb C/\Lambda^*; E ) $, we write $ L_\eta $ in terms of a bundle morphism:
\begin{equation}
\label{eq:defLeta}
(L_\eta s)(k_\eta) =\Lambda_\eta(k)s(k),  \ \ \
 \Lambda_\eta(k)
:=\Pi (k_\eta)T_\eta(k)|_{E_k}:
 E_k\longrightarrow E_{k_\eta}.
\end{equation}
A 
coordinate expression in a fundamental cell is given by 
$ \Lambda_\eta ( k )  =\Pi (k_\eta)  e^{-i\langle p_\eta(k),x\rangle}|_{E_k}$.
 We recall the properties of $ \Lambda_\eta ( k ) $: 
\begin{equation}
\label{eq:propLam}
 \Lambda_0(k)=I_{E_k},
 \qquad
 \Lambda_{-\eta}(k_\eta)=\Lambda_\eta(k)^*,
 \qquad
 \|\Lambda_\eta(k)\|\leq1.
\end{equation}
We also note that \(\Lambda_\eta(k)\) depends smoothly on \((k,\eta)\)
in the bundle sense.  Put
\[
 \kappa(k,\eta):=[k-\eta],
 \qquad
 \pi_1(k,\eta):=k.
\]
Then the maps \(T_\eta(k)\) in \eqref{eq:parentt} define a smooth bundle
isomorphism $
 T:\pi_1^*V\longrightarrow\kappa^*V $. 
(This is immediate from the quotient construction of \(V\),
since on the covering space \(T\) is induced by
$  (k,\eta,u)\longmapsto(k-\eta,u). $)
Since \(E\subset V\) and the orthogonal projections \(\Pi_E(k)\) are
smooth, \eqref{eq:defLeta} defines a smooth bundle morphism
\begin{equation}
\label{eq:Lamsm}
 \Lambda:\pi_1^*E\longrightarrow\kappa^*E.
\end{equation}
A representative \(k_\eta\) in a fixed fundamental cell need not
depend smoothly on \((k,\eta)\) but this is only an artefact of the
choice of representatives.

After applying ${\mathcal B}\otimes{\mathcal B}$, 
the interaction restricted to the Bloch subbundle $ E $  is
\begin{equation}
\begin{split}
 V_{E}:={}&
 ({\mathcal B}\otimes{\mathcal B})
 (\Pi_{E}\otimes\Pi_{E})V
 (\Pi_{E}\otimes\Pi_{E})
 ({\mathcal B}^{-1}\otimes{\mathcal B}^{-1})\\
 ={}&\frac1{(2\pi)^d}\int_{\mathbb C}
\widehat V (\eta)L_\eta\otimes L_{-\eta}d\eta,
\end{split}
\label{eq:projV}
\end{equation}
where $ \widehat V ( \eta ) := \int_{\mathbb R^2}
 V(x)e^{-i\langle\eta,x\rangle}d x$ is the Fourier transform of $ V$ -- 
see \cite[(7.14),(7.48)]{App}. 

Using the canonical splitting \eqref{eq:defPk} the above discussion applies to $ E $ replaced by line bundles $ E_\pm $
and in particular we have 
\begin{equation}
\label{eq:Lams}
\Lambda_\eta ( k ) = \begin{pmatrix} \Lambda_{+,\eta} (k) & \ \ 0 \\
\ \ 0 & \Lambda_{-,\eta}  (k) \end{pmatrix} . 
\end{equation}

\subsection{Hartree--Fock functional for flat bands at half-filling}
\label{s:HFmin}
We follow \cite[\S A.8]{App} and take potentials satisfying
\begin{equation}
\begin{gathered}
V \in C ( \mathbb R^2 ) , \ \ \  \widehat V(\eta)\geq0 ,  \ \ \  \widehat V ( \eta) = \mathcal O ( |\eta|^{-1} ) , \\
 \widehat V(\eta)>0 \text{ \ for \ } |\eta|<\varepsilon_V, \ \ \ \varepsilon_V > 0 . 
 \end{gathered} 
 \label{eq:positivity}
\end{equation}
However, we now assume a (possibly) {\em non-uniform} half-filling
 (as was also possible in \cite[\S A.8]{App} -- see \cite{SBL25} for the discrete case), that is 
 a measurable projection valued function, 
 \begin{equation}
 \Pi(k)^2=\Pi(k)=\Pi(k)^*,
 \ \ \ \ | \mathbb C/\Lambda^*|^{-1} \int_{\mathbb C/\Lambda^* } \tr  \Pi(k)= 1. 
\label{eq:P12}
\end{equation}
This is a {\em half-filling} as the projection to the full bundle satisfies $ \tr \Pi_E ( k ) = 2 $. 

Following the physics literature (see \cite[(3),(4), and (C8)--(C12)]{BSRL21}, 
\cite[\S I.3, \S IV.2, (15)--(17)]{F*23} and references given there) we consider the Hartree--Fock functional for a subtraction scheme taking into account interaction ``double counting''.
Its mathematical description is given in greater generality in 
\cite[(3.21)--(3.23)]{App} and here we recall from \cite[(3.31),(3.32)]{App} that we we minimising 
\begin{equation}
\label{eq:HFE} 
\Pi \mapsto E_{\rm{H}} ( \Pi ) + E_{\rm{X}}  ( \Pi ) ,
\end{equation}
where, putting
\begin{equation}
 Q(k):=2\Pi(k)-I_{E_k}, \ \  Q(k)^*=Q(k),\ \  Q(k)^2=\Id_{E_k},
 \ \ 
 |\mathbb C/\Lambda^*|^{-1}\int_{\mathbb C/\Lambda^*}\tr Q(k)\,dk=0. 
 \label{eq:defQ}
\end{equation}
we have
\begin{equation*}
\begin{gathered} 
\cEH(\Pi):=\frac{1}{4(2\pi)^2}|\mathbb C/\Lambda^*|^{-1}
 \sum_{p\in\Lambda^*}\widehat V(p)
 \left|\int_{\mathbb C/\Lambda^*}
 \tr\bigl(\Lambda_p(k)Q(k)\bigr)d k\right|^2, \\
 \cEX(\Pi): =-\frac{1}{4(2\pi)^2}|\mathbb C/\Lambda^*|^{-1}
 \int_{\mathbb C}\widehat V(\eta)
 \int_{\mathbb C/\Lambda^*}
 \tr\!\left(
 \Lambda_\eta(k)Q(k)\Lambda_\eta(k)^*Q(k_\eta)
 \right)d k\,d\eta.
\end{gathered}
\end{equation*}

The first theorem is an adaptation of the results of \cite{BLS25,SBL25} (for simple magic angles) to the
continuous setting. A more general version formulated using discrepancy of Chern numbers of
$ E_\pm $ was presented in \cite[\S A.8]{App} for uniform half-filling. For reader's convenience, we present 
a complete argument in the case of the chiral model of TBG:

\begin{thm}
\label{t:1}
Assume that $\widehat V$ satisfies
\eqref{eq:positivity}. 
If $ E = E_+ \oplus E_- $ where $ E_\pm $ are the flat band line bundles associated to a {\em simple} 
magic angle,  
the only measurable fields \eqref{eq:P12} minimising $ \Pi \to E_{\rm{H}}  ( \Pi ) + E_{\rm{X}} ( \Pi ) $ 
(see \eqref{eq:HFE}) are $\Pi=\Pi_+$ and $\Pi=\Pi_-$, where $ \Pi_\pm $ are the projection onto $ E_\pm $. 
\end{thm}
\begin{proof}
We start with the elementary identity, used earlier in this context in \cite{SBL25}:
\begin{equation}
 \tr(\Lambda Q\Lambda^*Q')
 =\|\Lambda\|_{\mathcal L_2}^2
 -\tfrac12\|Q'\Lambda-\Lambda Q\|_{\mathcal L_2}^2
 \label{eq:L2identity-new}
\end{equation}
which holds whenever $Q,Q'$ are self-adjoint involutions.  Hence
\begin{equation}
\begin{split}
 \cEX(\Pi)=\cEX^{\min}
 &+\frac{1}{8(2\pi)^2}|\mathbb C/\Lambda^*|^{-1}
 \int_{\mathbb C}\widehat V(\eta)
 \int_{\mathbb C/\Lambda^*}
 \|Q(k_\eta)\Lambda_\eta(k)-\Lambda_\eta(k)Q(k)\|_{\mathcal L_2}^2
 \,dk\,d\eta,
\end{split}
\label{eq:Focksquare-new}
\end{equation}
where
\begin{equation}
 \cEX^{\min}:=-\frac{1}{4(2\pi)^2}|\mathbb C/\Lambda^*|^{-1}
 \int_{\mathbb C}\widehat V(\eta)
 \int_{\mathbb C/\Lambda^*}
 \|\Lambda_\eta(k)\|_{\mathcal L_2}^2\,dk\,d\eta
 \label{eq:Fockmin-new}
\end{equation}
is independent of $\Pi$.  
All the quantities above are finite.  To see this, write
$\eta=q+p$, where $q$ belongs to a fixed fundamental domain of
$\Lambda^*$ and $p\in\Lambda^*$.  In local smooth orthonormal frames
of $E_k$ and $E_{k_q}$, the matrix entries of $\Lambda_{q+p}(k)$ have
the form
\[
 \int_{\mathbb C/\Lambda}
 e^{-i\langle p,x\rangle}a_{ij}(k,q,x)\,dx ,
\]
where $a_{ij}$ is smooth and periodic in $x$.  Since
$\mathbb C/\Lambda^*$ is compact, finitely many such frames suffice,
and all $x$-derivatives of the corresponding functions $a_{ij}$ are
bounded uniformly in $(k,q)$.  Repeated integration by parts therefore
gives, for every $N$,
\begin{equation}
 \sup_{k,q}
 \|\Lambda_{q+p}(k)\|_{\mathcal L_2(E_k,E_{k_q})}
 \leq C_N\langle p\rangle^{-N}.
 \label{eq:rapiddecay-new}
\end{equation}
Together with
$\widehat V(\eta)=\mathcal O(\langle\eta\rangle^{-1})$, this proves
the required absolute convergence.

Let $Q_\pm:=2\Pi_\pm-I_E$.  By \eqref{eq:Lams}, $Q_\pm$ commutes with
the bundle morphisms: 
\begin{equation}
 Q_\pm(k_\eta)\Lambda_\eta(k)=\Lambda_\eta(k)Q_\pm(k).
 \label{eq:componenttransport}
\end{equation}
We now check their Hartree energy.  For $p\in\Lambda^*$,
\eqref{eq:TBGsymmetries}--\eqref{eq:Asymmetry} give
\[
 \Lambda_p(k)=
 \begin{pmatrix}A_k(p)&0\\0&\overline{A_k(p)}\end{pmatrix},
 \qquad
 Q_+(k)=\begin{pmatrix}1&0\\0&-1\end{pmatrix},
\]
and therefore $
 \tr\bigl(\Lambda_p(-k)Q_+(-k)\bigr)
 =-\tr\bigl(\Lambda_p(k)Q_+(k)\bigr) $.
The change of variables $k\mapsto-k$ shows that
\begin{equation}
 \int_{\mathbb C/\Lambda^*}
 \tr\bigl(\Lambda_p(k)Q_+(k)\bigr)\,dk=0.
 \label{eq:neutrality-new}
\end{equation}
Since $Q_-=-Q_+$, the same holds for $Q_-$.  Thus
$\cEH(\Pi_\pm)=0$, and \eqref{eq:componenttransport} and
\eqref{eq:Focksquare-new} show that both $\Pi_+$ and $\Pi_-$ have energy
$\cEX^{\min}$.

Conversely, $\cEH\geq0$ and the second term in
\eqref{eq:Focksquare-new} is nonnegative.  Hence every minimizer must make
both terms vanish. 

We now adapt the proof of \cite[Theorem~8.1]{App} to
the non-uniform filling.
Hence, suppose that $Q$ is a minimizer of the Hartree-Fock energy. Since, $\widehat V>0$ on $|\eta|<\varepsilon_V$, 
Fubini's theorem gives
\begin{equation}
 Q(k_\eta)\Lambda_\eta(k)=\Lambda_\eta(k)Q(k)
 \label{eq:transport-new}
\end{equation}
for almost every $(k,\eta)$ with $|\eta|<\varepsilon_V$.

We next replace the merely measurable $Q$ by a smooth representative.
Since $\Lambda_0(k)=I_{E_k}$, smoothness and compactness give
$\varepsilon_0>0$ such that
\begin{equation}
 \Lambda_\eta(k)^*\Lambda_\eta(k)\geq\tfrac12I_{E_k},
 \qquad |\eta|<\varepsilon_0.
 \label{eq:invertibility-new}
\end{equation}
Put $\varepsilon_*:=\min(\varepsilon_V,\varepsilon_0)$ and choose
$\chi\in C_c^\infty(B(0,\varepsilon_*))$, $\chi\geq0$, with
$\int\chi(\eta)\,d\eta=1$.  Then
\begin{equation}
 \widetilde Q(k):=\int\chi(\eta)\Lambda_\eta(k)^{-1}
 Q(k_\eta)\Lambda_\eta(k)\,d\eta
 \label{eq:mollify-new}
\end{equation}
equals $Q(k)$ for almost every $k$, by \eqref{eq:transport-new}, and is
smooth.  This follows in local trivializations, after the change of
variables $k'=k_\eta$, since \eqref{eq:mollify-new} is then the integral of
the $L^\infty$ field $Q(k')$ against a smooth compactly supported
kernel (see \eqref{eq:Lamsm}).  From now on we use this smooth representative.  The two sides of
\eqref{eq:transport-new} are now smooth in $(k,\eta)$, so that the
identity holds for every $k$ and every $|\eta|<\varepsilon_*$.

We can therefore assume that $\Pi=(Q+I_E)/2$ is a smooth family of
projections.  Its rank is therefore constant on the connected torus
$\mathbb C/\Lambda^*$.  The non-uniform half-filling condition
\eqref{eq:P12}, or equivalently the last identity in \eqref{eq:defQ},
then forces this constant rank to be one.  Hence
\[
 L:=\Ran\Pi\subset E_+\oplus E_-
\]
is a smooth line subbundle.
By \eqref{eq:TBGchern},
$c_1(E_+)\ne c_1(E_-)$.  Thus at some $k_0$ one has
\begin{equation}
 L_{k_0}=E_{+,k_0}\qquad\text{or}\qquad L_{k_0}=E_{-,k_0}.
 \label{eq:onepoint-new}
\end{equation}
Indeed, otherwise the two bundle projections $L\to E_+$ and $L\to E_-$
would be nowhere zero, hence line-bundle isomorphisms, which would imply
$c_1(E_+)=c_1(L)=c_1(E_-)$.

Finally, \eqref{eq:transport-new}, \eqref{eq:invertibility-new}, and
\eqref{eq:componenttransport} imply, for $|\eta|<\varepsilon_*$,
\[
 \Lambda_\eta(k)L_k=L_{k_\eta},
 \qquad
 \Lambda_\eta(k)E_{\pm,k}=E_{\pm,k_\eta}.
\]
Either equality in \eqref{eq:onepoint-new} therefore propagates to a
neighborhood.  Its equality set is also closed, and the torus
$\mathbb C/\Lambda^*$ is connected.  Hence $L=E_+$ or $L=E_-$, that is,
$\Pi=\Pi_+$ or $\Pi=\Pi_-$ almost everywhere.
\end{proof}

\section{Strict coercivity of the Hartree--Fock functional}
\label{s:scHF}

In this section we prove that the functional \eqref{eq:HFE} has a strictly positive Hessian 
at the minimisers obtained in Theorem \ref{t:1}.  Although that theorem allowed non-uniform filling, 
both minimisers have rank one at every $k$.  We therefore
consider the Hessian on the Hilbert manifold of rank-one projection fields.  The proof of
coercivity follows from analysing \eqref{eq:L2identity-new} and a simple topological argument.

\subsection{Tangent space}
\label{s:tang}

For the purposes of this section, choose a unitary trivialisation
\[
        E\simeq (\mathbb C/\Lambda^*)\times\mathbb C^2.
\]
The rank-one projection fields of $H^2$ regularity form a Hilbert manifold, identified in this
trivialisation with
$ H^2(\mathbb C/\Lambda^*;\mathbb S^2) $
by
\begin{equation}
\label{eq:defPk1}
   \Pi_f(k)=\tfrac12\big(I+f(k)\cdot\sigma\big),
   \qquad
   Q_f(k):=2\Pi_f(k)-I=f(k)\cdot\sigma,
\end{equation}
where $\sigma=(\sigma_1,\sigma_2,\sigma_3)$ are the Pauli matrices.  We stress that the
trivialisation is used only to write this formula.  We also stress that $ H^2 $ could be replaced by $ H^r $ for any $ r > 1 $ (so that the resulting functions of $ k $ are continuous) but we fix a Hilbert space giving 
a natural manifold structure to the projections.

The canonical splitting
\begin{equation}
\label{eq:cano}
        E=E_+\oplus E_-
\end{equation}
is {\em not} the fixed coordinate splitting of $\mathbb C^2$ in this trivialisation.

Let $f_+\in\mathcal M$ correspond to $\Pi_+$ in the identification \eqref{eq:defPk1}. Then 
$        Q_+ =2\Pi_+-I=f_+\cdot\sigma$. 
Since $\Pi_+$ is the projection onto $E_+$, the $+1$ and $-1$ eigenspace bundles of
$Q_+$ are $E_+$ and $E_-$,  respectively. Thus the canonical splitting \eqref{eq:cano} is
the eigenspace splitting of the generally non-constant matrix field $Q_+(k)$.

The tangent space at $f_+$ is
\begin{equation}
\label{eq:deftan}
\begin{split}
 \mathcal T_+
 &: =  T_{f_+} \left( H^2 ( \mathbb C/\Lambda^*; \mathbb S^2 ) \right)  
 =\left\{\varphi\in H^2(\mathbb C/\Lambda^*;\mathbb R^3):
 \ \ \forall \, k \ \langle\varphi(k),f_+(k)\rangle_{\mathbb R^3}=0
\right\}                                          \\
 &=H^2(\mathbb C/\Lambda^*;f_+^*T \mathbb S^2).
\end{split}
\end{equation}
(Here $ f_+^* T \mathbb S^2 $ is the pullback for the tangent bundle $ T \mathbb S^2 
\to \mathbb S^2 $ -- see \cite[(2.29)]{TaZw}.)
The orthogonality condition in \eqref{eq:deftan} follows by differentiating
$|f_t(k)|^2=1$.  Conversely, every $\varphi\in\mathcal T_+$ is obtained
 by differentiating a $ C^2 $ curve in
$ t \mapsto f_t \in H^2(\mathbb C/\Lambda^*;\mathbb S^2)$, $ f_0 = f_+ $, at $ t = 0 $. We can 
take for instance, 
\begin{equation}
\label{eq:curvef}
 f_t(k)=\frac{f_+(k)+t\varphi(k)}
 {(1+t^2|\varphi(k)|^2)^{1/2}}.
\end{equation}
For $\varphi\in\mathcal T_+$ put
\[
        X_\varphi(k):=\varphi(k)\cdot\sigma.
\]
Then
\begin{equation}
 X_\varphi^*=X_\varphi,\ \ \ 
 Q_+X_\varphi+X_\varphi Q_+=0,\ \ \ 
 \|X_\varphi(k)\|_{\mathcal L_2}^2=2|\varphi(k)|^2.
 \label{eq:Xphinorm}
\end{equation}
The anticommutation identity in \eqref{eq:Xphinorm} says exactly that
$X_\varphi$ is off-diagonal with respect to the canonical splitting
$E=E_+\oplus E_-$.  Indeed, since
\[
 \Pi_+=\tfrac12(I+Q_+), \ \ \ 
 \Pi_-=\tfrac12(I-Q_+),
\]
the identity $Q_+X_\varphi+X_\varphi Q_+=0$ and  $Q_-=-Q_+$ give 
$  \Pi_+X_\varphi\Pi_+ =  \Pi_-X_\varphi\Pi_-=0$.
Consequently,
\[
 X_\varphi=\psi_\varphi+\psi_\varphi^*, \ \ \ 
 \psi_\varphi:=\Pi_-X_\varphi\Pi_+
 \in H^2\big(\mathbb C/\Lambda^*;\operatorname{Hom}(E_+,E_-)\big).
\]

Thus $f_+^*T\mathbb S^2$ is identified, as a real rank-two bundle, with the underlying real bundle
of $\operatorname{Hom}(E_+,E_-)$.  In a local unitary frame adapted to the canonical
splitting, rather than in the fixed trivialisation used to define the Pauli matrices, this says
simply that
\[
 Q_+(k)=
 \begin{pmatrix}1& \ \ 0\\ 0&-1\end{pmatrix},\ \ \ \ 
 X_\varphi(k)=
 \begin{pmatrix}0&\overline{z_\varphi(k)}\\z_\varphi(k)&0\end{pmatrix}.
\]
The scalar $z_\varphi$ is only a local coordinate; globally the lower-left entry is the section
$\psi_\varphi$ above.

\subsection{The Hessian at the minimisers}
We now prove coercivity on the $L^2$ completion of $ \mathcal T_+ $ in \eqref{eq:deftan}
$ L^2(\mathbb C/\Lambda^*;f_+^*T\mathbb S^2) $. 
For $\varphi\in\mathcal T_+$, choose a $C^2$ curve
\[
 t\longmapsto f_t\in H^2(\mathbb C/\Lambda^*;\mathbb S^2),
 \qquad f_0=f_+,
 \qquad \partial_tf_t|_{t=0}=\varphi,
\]
for instance the curve in \eqref{eq:curvef}.  Using that we define the Hessian as the quadratic form
\begin{equation}
 \operatorname{Hess}_{\Pi_+}\cEHF(\varphi,\varphi)
 :=\left.\frac{d^2}{dt^2}\right|_{t=0}\cEHF(\Pi_{f_t}), \ \ \ \Pi_+ = \Pi_{f_+} . 
 \label{eq:hessdef}
\end{equation}
As in finite dimensions the Hessian is independent of the choice of $ f_t $ as
$ \Pi_{f_+} $ is critical for $ \cEHF ( \Pi_{f_t } ) $.  Our main result is the strict coercivity of this
Hessian:

\begin{thm}
\label{t:2}
In the setting of Theorem~\ref{t:1}, the Hessian of $\cEHF$ restricted to the rank-one
Hilbert manifold at $\Pi_+$ extends from $\mathcal T_+$ to a bounded quadratic form on
$\mathcal H_+$.  There exists $c_+>0$ such that
\begin{equation}
\operatorname{Hess}_{\Pi_+}\cEHF(\varphi,\varphi)
 \geq c_+\int_{\mathbb C/\Lambda^*}|\varphi(k)|^2\,dk,
 \qquad \varphi\in \mathcal H_+ :=   L^2(\mathbb C/\Lambda^*;f_+^*T\mathbb S^2).
 \label{eq:hesscoercive}
\end{equation}
The analogous statement holds at $\Pi_-$.
\end{thm}

\begin{proof}

By \eqref{eq:neutrality-new},
$\cEH(\Pi_+)=0$, while $\cEH\geq0$.  Hence the second variation of
$\cEH$ at $\Pi_+$ along the rank-one Hilbert manifold is nonnegative, and it is enough
to prove a coercive lower bound for the exchange Hessian.

The square in \eqref{eq:Focksquare-new} vanishes at $Q_+$, and differentiation gives
the following expression for \eqref{eq:hessdef}:
\begin{equation}
\frac{1}{4(2\pi)^2}|\mathbb C/\Lambda^*|^{-1}
 \int_{\mathbb C}\widehat V(\eta)
 \int_{\mathbb C/\Lambda^*}
 \big\|X_\varphi(k_\eta)\Lambda_\eta(k)
       -\Lambda_\eta(k)X_\varphi(k)\big\|_{\mathcal L_2}^2
 \,dk\,d\eta.
\label{eq:hessformula}
\end{equation}
Thus the exchange Hessian has the same transfer-defect form as the positive part of the
exchange energy.  This is because that part of the energy is the square of an expression
which is linear in $Q$ and vanishes at the minimiser.  The same estimates as in
\eqref{eq:rapiddecay-new} show that the right hand side of \eqref{eq:hessformula}
extends to a bounded quadratic form on $ \mathcal H_+ := L^2(\mathbb C/\Lambda^*;f_+^*T\mathbb S^2) $, 
and
\begin{equation}
 \operatorname{Hess}_{\Pi_+}\cEHF(\varphi,\varphi)
 \geq \operatorname{Hess}_{\Pi_+}\cEX(\varphi,\varphi),
 \qquad \varphi\in\mathcal T_+.
 \label{eq:HFgeqX}
\end{equation}

We recall from \eqref{eq:invertibility-new} that $ 
 \Lambda_\eta(k)^*\Lambda_\eta(k)\geq\tfrac12 I$,
for $ |\eta|<\epsilon $,
For those values,  let
\begin{equation}
 \Lambda_\eta(k)=U_\eta(k)S_\eta(k),\ \ \ 
 S_\eta(k):=\bigl(\Lambda_\eta(k)^*\Lambda_\eta(k)\bigr)^{1/2},
 \label{eq:polar}
\end{equation}
be its polar decomposition.  Then $U_\eta(k):E_k\to E_{k_\eta}$ is unitary and
$S_\eta(k)\geq 2^{-1/2}I$.  Since $\Lambda_\eta(k)$ preserves the canonical splitting,
$S_\eta(k)$ preserves $E_{+,k}$ and $E_{-,k}$, while $U_\eta(k)$ maps these spaces
unitarily onto $E_{+,k_\eta}$ and $E_{-,k_\eta}$, respectively.  Both factors depend
smoothly on $(k,\eta)$.

The identity \eqref{eq:componenttransport} at $Q_+$ and the polar decomposition give
\[
 Q_+(k_\eta)U_\eta(k)=U_\eta(k)Q_+(k).
\]
Indeed, $S_\eta(k)$ commutes with $Q_+(k)$ and is invertible.  Conjugation by
$U_\eta(k)$ defines $R_\eta(k)\in SO(3)$ by
\begin{equation}
 U_\eta(k)(v\cdot\sigma)U_\eta(k)^*
 =\bigl(R_\eta(k)v\bigr)\cdot\sigma,
 \qquad v\in\mathbb R^3.
 \label{eq:Retadef}
\end{equation}
Since $Q_+(k)=f_+(k)\cdot\sigma$, the intertwining identity gives
\[
        R_\eta(k)f_+(k)=f_+(k_\eta).
\]
Consequently,
\[
 R_\eta(k):T_{f_+(k)}\mathbb S^2\longrightarrow T_{f_+(k_\eta)}\mathbb S^2
\]
is an isometry, depending smoothly on $(k,\eta)$, and
$R_{-\eta}(k_\eta)=R_\eta(k)^{-1}$.

In particular,
$
 U_\eta(k)^*X_{\varphi(k_\eta)}U_\eta(k)
 =X_{R_\eta(k)^{-1}\varphi(k_\eta)}$. 
Since $R_\eta(k)^{-1}\varphi(k_\eta)$ and $\varphi(k)$ both belong to
$T_{f_+(k)}\mathbb S^2$, the two self-adjoint matrices
\[
 A:=U_\eta(k)^*X_{\varphi(k_\eta)}U_\eta(k),\ \ \ 
 B:=X_{\varphi(k)}
\]
are off-diagonal with respect to the same canonical splitting
$E_k=E_{+,k}\oplus E_{-,k}$.  

We now use the following elementary
Hilbert--Schmidt norm estimate.  In a local frame adapted to this splitting, write
$S=\operatorname{diag}(s_+,s_-)$ and represent the off-diagonal self-adjoint
matrices $A,B$ by $a,b\in\mathbb C$.  Then
\[
    \|AS-SB\|_{\mathcal L_2}^2
    =2s_+s_-|a-b|^2+(s_+-s_-)^2(|a|^2+|b|^2), 
\]
and since $\|A-B\|_{\mathcal L_2}^{2} = 2 | a - b |^{2}$, 
\begin{equation}
 S\geq\gamma I\ \Longrightarrow\ 
 \|AS-SB\|_{\mathcal L_2}\geq\gamma\|A-B\|_{\mathcal L_2}.
 \label{eq:matrixineq}
\end{equation}
(For related inequalities in greater generality, see for instance,
\cite[Chapter~VIII.3]{Bh97}.)

Applying \eqref{eq:matrixineq} with $S=S_\eta(k)$ and using
\eqref{eq:Xphinorm}, we obtain (note that \eqref{eq:Xphinorm} has a factor of $ 2 $)
\begin{equation}
 \big\|X_\varphi(k_\eta)\Lambda_\eta(k)
       -\Lambda_\eta(k)X_\varphi(k)\big\|_{\mathcal L_2}^2
 \geq
 \big|\varphi(k_\eta)-R_\eta(k)\varphi(k)\big|^2.
 \label{eq:covdiffbound}
\end{equation}
To analyse the right hand side, we 
define the unitary operator $\mathcal U_\eta$ on $\mathcal H_+$ by
\begin{equation}
 (\mathcal U_\eta\varphi)(k_\eta):=R_\eta(k)\varphi(k).
 \label{eq:Ueta}
\end{equation}
The identity $\Lambda_{-\eta}(k_\eta)=\Lambda_\eta(k)^*$ and uniqueness of the polar
decomposition give $\mathcal U_{-\eta}=\mathcal U_\eta^*$.  Shrinking $\epsilon$ below the
injectivity radius of the torus if necessary, choose an even function
$\chi\in C_c^\infty(B(0,\epsilon))$, $\chi\geq0$, which is positive near $0$, and put
\begin{equation}
 q(\varphi):=\int_{\mathbb C}\chi(\eta)
 \|\mathcal U_\eta\varphi-\varphi\|_{\mathcal H_+}^2\,d\eta.
 \label{eq:qdef}
\end{equation}
It follows from \eqref{eq:HFgeqX}, \eqref{eq:hessformula}, and \eqref{eq:covdiffbound} that
\begin{equation}
 \operatorname{Hess}_{\Pi_+}\cEHF(\varphi,\varphi)\geq c_0q(\varphi)
 \label{eq:hessq}
\end{equation}
for some $c_0>0$.  It remains to prove that
\begin{equation}
 q(\varphi)\geq c_1\|\varphi\|_{\mathcal H_+}^2,\ \ \ c_1>0.
 \label{eq:qcoer}
\end{equation}
For that we define
\begin{equation}
 c_\chi:=\int_{\mathbb C}\chi(\eta)\,d\eta>0,\ \ \ 
 K_\chi:=\int_{\mathbb C}\chi(\eta)\mathcal U_\eta\,d\eta:
  \mathcal H_+  \longrightarrow  \mathcal H_+.
 \label{eq:Kchi}
\end{equation}
Since $\mathcal U_\eta$ is unitary and $\chi$ is even,
$K_\chi=K_\chi^*$, $\|K_\chi\|\leq c_\chi$, and
\begin{equation}
 q(\varphi)=2c_\chi\|\varphi\|_{\mathcal H_+ }^2
 -2 \Re \, \langle K_\chi\varphi,\varphi\rangle_{\mathcal H_+}.
 \label{eq:qKchi}
\end{equation}
The operator $K_\chi$ is smoothing, hence compact.  Indeed, after the change of variables
$k'=k_\eta$, it is an integral operator with a smooth bundle kernel, since
$(k,\eta)\mapsto R_\eta(k)$ is smooth and $\chi$ is smooth and supported in a sufficiently
small neighborhood of $0$.

Suppose that \eqref{eq:qcoer} were false.  We could then find
$\varphi_j\in\mathcal H_+$ such that
\[
 \|\varphi_j\|_{\mathcal H_+}=1,\ \ \ q(\varphi_j)\longrightarrow0.
\]
After passage to a subsequence, $\varphi_j\rightharpoonup\varphi$ weakly in
$\mathcal H_+$, with $\|\varphi\|_{\mathcal H_+}\leq1$.  Compactness gives
$K_\chi\varphi_j\to K_\chi\varphi$ strongly, and \eqref{eq:qKchi} therefore yields
\[
      \Re \, \langle K_\chi\varphi,\varphi\rangle_{\mathcal H_+}=c_\chi.
\]
Since $\|K_\chi\|\leq c_\chi$, this implies
$\|\varphi\|_{\mathcal H_+}=1$.  Hence $\varphi_j\to\varphi$ strongly and
$q(\varphi)=0$.

The nonnegative operator $c_\chi I-K_\chi$ therefore annihilates $\varphi$.  Since
$K_\chi$ is smoothing, $\varphi$ is smooth.  From \eqref{eq:qdef} and the fact that
$\chi$ is positive near $0$, we obtain
\begin{equation}
 \varphi(k_\eta)=R_\eta(k)\varphi(k),\ \ \ 
 k\in\mathbb C/\Lambda^*,\ \ \ |\eta|\ll1.
 \label{eq:parallelphi}
\end{equation}
The zero set of $\varphi$ is consequently both open and closed.  Since the torus is
connected and $\|\varphi\|_{\mathcal H_+}=1$, the section $\varphi$ is nowhere zero.  On
the other hand,
\[
 f_+^*T\mathbb S^2\simeq\operatorname{Hom}(E_+,E_-),\ \ \ 
 c_1\bigl(\operatorname{Hom}(E_+,E_-)\bigr)=2.
\]
Thus $f_+^*T\mathbb S^2$ has no nowhere-vanishing sections, which gives a contradiction.  Hence
\eqref{eq:qcoer} holds, and \eqref{eq:hessq} proves \eqref{eq:hesscoercive}.

The assertion at $\Pi_-$ follows by interchanging $E_+$ and $E_-$.  Equivalently,
$Q_-=-Q_+$ and $f_-=-f_+$, so the same tangent planes and the same argument apply; the
relevant complex line bundle is $\operatorname{Hom}(E_-,E_+)$, whose Chern number is
$-2$. 
\end{proof}

\noindent
{\bf Remark.} 
Although only positivity of $\cEH$ was needed in the proof, its second variation at the
two minimisers actually vanishes.  Indeed, put
\[
 F_p(Q):=\int_{\mathbb C/\Lambda^*}
 \tr\bigl(\Lambda_p(k)Q(k)\bigr)\,dk.
\]
Then $F_p(Q_+)=0$ by \eqref{eq:neutrality-new}.  If $Q_t$ is a curve through $Q_+$
with $\partial_tQ_t|_{t=0}=X_\varphi$, then
\[
 \left.\frac{d}{dt}\right|_{t=0}F_p(Q_t)
 =\int_{\mathbb C/\Lambda^*}
 \tr\bigl(\Lambda_p(k)X_\varphi(k)\bigr)\,dk=0,
\]
since $\Lambda_p(k)$ is diagonal and $X_\varphi(k)$ is off-diagonal with respect
to $E_+\oplus E_-$.  Hence
\[
 \left.\frac{d^2}{dt^2}\right|_{t=0}|F_p(Q_t)|^2=0,
\]
and consequently, using the same absolute convergence estimates as above,
\[
 \operatorname{Hess}_{\Pi_+}\cEH(\varphi,\varphi)=0.
\]
Thus the full Hartree--Fock Hessian is exactly the exchange Hessian in \eqref{eq:hessformula}.

\end{document}